\documentclass{ifacconf}
\ifdefined\pdfminorversion
\else
  
\fi

\usepackage[utf8]{inputenc}
\usepackage[T1]{fontenc}
\usepackage{amsmath,amssymb,amsfonts}
\usepackage{graphicx}
\usepackage{natbib}
\usepackage{textcomp}
\usepackage{url}
\usepackage{tikz}
\usepackage{dsfont}
\let\mathbbm\mathds
\usepackage{newunicodechar}
\newunicodechar{—}{\textemdash}
\usetikzlibrary{arrows.meta,fit,backgrounds}
\tikzstyle{block} = [draw, rectangle,
    minimum height=3em, minimum width=3em]
\tikzstyle{sum} = [draw, circle, node distance=1cm]
\tikzstyle{input} = [coordinate]
\tikzstyle{output} = [coordinate]
\tikzstyle{pinstyle} = [pin edge={to-,thin,black}]
\graphicspath{{./figures/}}
\newtheorem{theorem}{Theorem}

\newtheorem{remark}{Remark}

\newenvironment{proof}[1][Proof]
{\par\noindent\textit{#1.}\ }
{\hfill$\square$\par}

\newcommand{\BBM}{\begin{bmatrix}}
\newcommand{\EBM}{\end{bmatrix}}
\newcommand{\BEQ}{\begin{equation}}
\newcommand{\EEQ}{\end{equation}}
\newcommand{\BIT}{\begin{itemize}}
\newcommand{\EIT}{\end{itemize}}

\newcommand{\st}{\mbox{subject to}}

\newcommand{\reals}{\mathbb{R}}

\newcommand{\minimize}{\mbox{minimize}}

\DeclareMathOperator{\diag}{diag}
\DeclareMathOperator{\blkdiag}{blkdiag}

\let\cite\citep
\begin{document}
\begin{frontmatter}
\title{Prediction Error Method for Identification of Subnetworks of Dynamical Systems\thanksref{funding}}
\thanks[funding]{This work was supported by ELLIIT, and by the Novo Nordisk Foundation under grant number NNF20OC0061894.}
\author[Lund]{João Victor Galvão da Mata}
\author[LiU]{Anders Hansson}
\author[DTU]{Martin S.~Andersen}
\address[LiU]{Anders Hansson is with the Department of Electrical Engineering, Linköping University (e-mail: anders.g.hansson@liu.se).}
\address[Lund]{João Victor Galvão da Mata is with the Centre for Mathematical Sciences, Lund University (e-mail: victor.mata@math.lu.se).}
\address[DTU]{Martin S.~Andersen is with the Department of Applied Mathematics and Computer Science, Technical University of Denmark (e-mail: mskan@dtu.dk).}
\begin{abstract}
This paper considers prediction-error identification of a subnetwork embedded
in a dynamic network. Measured boundary signals entering the subnetwork are
appended to the exogenous inputs. Consistency requires that the part of these
boundary signals generated by the subnetwork be reconstructible from the
measured history; the subnetwork may be enlarged to obtain a simple
structural sufficient condition. Under the stated model, stability, delay,
disturbance-separation, informativity, and identifiability assumptions, the
resulting criterion yields consistent estimates of the subnetwork transfer
functions. The predictor may depend
nonlinearly on the open-loop transfer functions, allowing consistent
estimation with fewer measured internal signals than methods based on
predictors that are linear in those transfer functions.
\end{abstract}
\end{frontmatter}
\section{Introduction}
System identification for networks of dynamical systems is an active research 
field with a 
long history. Some early references 
are \cite{VANDENHOF19931523,VanDenHof2013,hen+gev+baz19}. Much of the focus has been on 
identifying parts of a network, called a {\it subnetwork}, with the motivation that in practice it is impossible
to identify large  networks, and often only parts of a network are of interest in applications. 
However, it is important to recognize
that the overall network will have an impact on how the subnetwork can be identified and what the quality of the
subnetwork model will be. This impact is due to the feedback that the overall network usually provides 
to the subnetwork. This means that the problem of identifying subnetworks is a problem of system 
identification under feedback, a topic that also has a long history, see e.g. \cite{GUSTAVSSON197759}. A more recent survey is \cite{Forssell1999}.

We develop a prediction error method (PEM) for this problem. The main idea is
to treat measured boundary signals entering the target subnetwork as
additional measured inputs and then apply complete-network PEM.

Identification of individual modules and subnetworks in linear dynamic
networks has received considerable attention. The local methods of
\cite{VanDenHof2013} construct predictors around a designated module, and
\cite{Dankers2016} give conditions for selecting subsets of node signals as
predictor inputs. These methods target one designated transfer function, and
their local predictor formulations are linear in that target open-loop module.
Network abstraction provides another way to remove unmeasured node signals
while preserving a designated target module \cite{Weerts2020}.
Complete-network methods instead estimate all modules jointly
\cite{Weerts2018}. More generally,
\cite{hansson2025identifiabilitymaximumlikelihoodestimation} establish
consistency for complete dynamic networks with predictors that may depend
nonlinearly on several open-loop transfer functions. This result provides the
complete-network basis for the subnetwork method developed here.

In the standard closed-loop setting considered by
\cite{Ljung1999,Forssell1999}, the output entering the feedback loop is also
the output being predicted. Their prediction-error analysis assumes that the
plant contribution to the predicted output is independent of the current and
future innovations. A delay in the loop is a common sufficient condition: the
current prediction error cannot pass through the feedback loop and affect the
input used to predict the same output.

This property is not automatic when only part of a dynamic network is
measured. A boundary signal returning to the subnetwork may depend on an
unmeasured subnetwork output and can therefore contain information that is
absent from the measured output history. Existing local-module methods address
this issue by enlarging the local predictor input and output sets so that the
relevant feedback paths are accounted for
\cite{VanDenHof2013,Dankers2016,Weerts2020}.

Our method uses a different enlargement rule. Additional modules from the
original network are included in the subnetwork to be estimated until the part
of each returning boundary signal generated by the subnetwork is measured or
reconstructible from the measured history. All modules in the enlarged
subnetwork are then estimated jointly, and the prediction-error orthogonality
argument can be applied. Our work builds on
\cite{hansson2025identifiabilitymaximumlikelihoodestimation}, which considers
complete dynamic networks. If only one transfer function is of interest, the
initial subnetwork may consist of its module, so single-module identification
is included as a special case.

Because the complete-network predictor can depend nonlinearly on several
unknown open-loop transfer functions, the interactions among modules in the
enlarged subnetwork can be represented directly. This can reduce the number of
internal signals that must be measured compared with local methods based on
predictors that are linear in the target modules.
\section{Model of Network}
We consider linear models for the vector-valued input signals 
$u_k\in\reals^n$, $r_k\in\reals^{m_r}$ and 
$e_k\in\reals^{m_e}$, and the vector-valued output signal $y_k\in\reals^n$ 
defined by
\BEQ\label{eqn:y-model}
y_k = \mathbf G(q)u_k + \mathbf \Omega_y(q)r_k+\mathbf H_y(q) e_k
\EEQ
where $\mathbf G$, $\mathbf \Omega_y$ and $\mathbf H_y$ are proper
transfer function
matrices in the forward shift operator,
$q$. We assume that $\mathbf G(q)$ is strictly proper. The sequence $e_k$ is
temporally independent, has zero mean, and has constant diagonal positive definite
covariance matrix $\Sigma_e$.
Let $\mathbf \Upsilon$, $\mathbf \Omega_u$ and $\mathbf H_u$, of dimensions
$n\times n$, $n\times m_r$ and $n\times m_e$, respectively, also be proper
transfer function matrices in the forward shift operator. 
We define a 
feedback connection 
\BEQ\label{eqn:u-model}
u_k=\mathbf \Upsilon(q) y_k+\mathbf \Omega_u(q) r_k + \mathbf H_u(q)e_k.
\EEQ
These two equations 
define what is often called a {\it network} of dynamical 
systems. We assume that the 
closed-loop system is well-posed and stable.

Another popular network formulation considered in the literature is given by
\BEQ\label{eqn:w-model}
w_k=\bar{\mathbf G}(q)w_k+\bar{\mathbf \Omega}(q)r_k+\bar{\mathbf H}(q)e_k
\EEQ
with $(\bar{\mathbf G}, \bar{\mathbf\Omega},\bar{\mathbf H})$ also 
proper transfer function matrices in the forward shift operator. We assume that
$\bar{\mathbf H}(q)$ is diagonal.
It is often assumed that $\bar{\mathbf G}(q)$ has a zero diagonal.
This can be cast in the form of \eqref{eqn:y-model} by taking $u_k=y_k=w_k$ and 
$\mathbf G(q)=\bar{\mathbf G}(q)$. A network on the form \eqref{eqn:y-model} and
\eqref{eqn:u-model} can  be written
\BEQ\label{eqn:yutow}
\BBM y_k\\u_k\EBM=\BBM 0&\mathbf G(q)\\\mathbf \Upsilon(q)&0\EBM \BBM y_k\\u_k\EBM+
\BBM\mathbf \Omega_y(q)\\\mathbf \Omega_u(q)\EBM r_k+
\BBM\mathbf H_y(q)\\\mathbf H_u(q)\EBM e_k
\EEQ
which with  $w_k=(y_k,u_k)$ and obvious definitions of transfer function matrices is of the 
form in \eqref{eqn:w-model} with a zero diagonal $\bar{\mathbf G}(q)$.
Hence the two formulations can be considered to be equivalent, and we
will use them interchangeably. 

We  assume that all the transfer functions
we have defined are dependent on some parameter $\theta$, parts of which 
are going to be identified. 
We denote with $\Theta_o$ the set 
of parameters $\theta$ such that the
closed-loop system has all poles strictly inside the unit circle. 
We define the compact set $\Theta\subset \Theta_o$ such that it contains the parameter
$\theta_0$ which represents the closed-loop system from which we will collect data. This data will be defined 
later on. 
We assume that all transfer functions are rational functions in
$\theta$ and $q$. If we want to emphasize the parameter dependence we may e.g.
write $\mathbf G(q;\theta)$ for transfer functions and $y_k(\theta)$ for signals. 
\section{Prediction Error Method and Consistency}
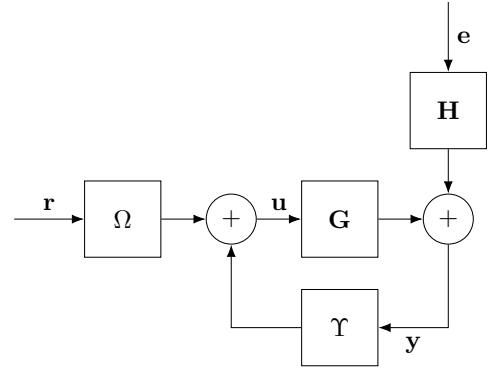
\begin{figure}[htbp]
\begin{center}
\resizebox{0.35\textwidth}{!}{\begin{tikzpicture}[auto, node distance=2cm]

\node [input, name=r] {};
\node [block, right of=r, node distance=1.5cm] (Omega) {$\Omega$};
\node [sum, right of=Omega, node distance=1.5cm] (sum1) {$+$};
\node [block, right of=sum1, node distance=1.5cm] (G) {$\mathbf G$};
\node [sum, right of=G, node distance=1.5cm] (sum2) {$+$};
\node [block, above of=sum2, node distance=1.5cm] (H) {$\mathbf H$};
\node [input, above of=H, node distance=1.5cm] (e) {};

\node [input, below of=sum1, node distance=1.5cm] (aux1) {};
\node [input, below of=sum2, node distance=1.5cm] (aux2) {};

\node [block, below of=G, node distance=1.5cm] (Upsilon) {$\Upsilon$};

 \draw [-Latex] (r) -- node {$\mathbf r$} (Omega); 
\draw [-Latex] (Omega) -- node {} (sum1); 
\draw [-Latex] (sum1) -- node {$\mathbf u$} (G); 
\draw [-Latex] (G) -- node {} (sum2);
\draw [-Latex] (H) -- node {} (sum2);
\draw [-Latex] (e) -- node {$\mathbf e$} (H); 

\draw [-Latex] (sum2.south) to (aux2)  -- node {$\mathbf y$} (Upsilon);
\draw [-Latex] (Upsilon) to (aux1)  -- node {} (sum1.south);

\end{tikzpicture}}
\end{center}
\caption{Block diagram for a dynamic network.}
\centering
\label{fig:block-overall}
\end{figure}
It is possible to consider a problem equivalent to (\ref{eqn:y-model}--\ref{eqn:u-model})
which is of the same form but of higher dimension for which 
$\mathbf G$  is diagonal, $\Upsilon$ and $\mathbf \Omega_u$ are
zero-one matrices and $\mathbf\Omega_y$ and $\mathbf H_u$ are zero. This is shown in Section~\ref{sec:appendix}.
For that case we have for diagonal $\mathbf H=\mathbf H_y$ 
a multidimensional Box-Jenkins model in a feedback configuration as in 
Figure~\ref{fig:block-overall}.
In \cite{hansson2025identifiabilitymaximumlikelihoodestimation} it has been shown how to obtain the predictor and formulate a maximum likelihood (ML) problem
that can be solved reliably for multiple single-input
single-output ARMAX models in this feedback configuration
also when not all signals in the network are measured. 

To estimate part of a network, partition it as
\begin{align*}\BBM w^A_k\\w^B_k\EBM&=\BBM \bar{\mathbf G}_A&\bar{\mathbf G}_{AB}\\\bar{\mathbf G}_{BA}&\bar{\mathbf G}_{B}\EBM
\BBM w^A_k\\w^B_k\EBM+\BBM \bar{\mathbf\Omega}_A\\\bar{\mathbf \Omega}_B\EBM r_k+
\BBM \bar{\mathbf H}_A\\\bar{\mathbf H}_B\EBM e_k.
\end{align*}
The signals and modules with superscript A form the A-part, and those with
superscript B form the B-part.

We define $\hat r_k^A=(r_k,w_k^B)$, $\hat r_k^B=(r_k,w_k^A)$ and 
$$\hat{\mathbf \Omega}_A=\BBM \bar{\mathbf\Omega}_A&\bar{\mathbf G}_{AB}\EBM,\quad
\hat{\mathbf \Omega}_B=\BBM \bar{\mathbf\Omega}_B&\bar{\mathbf G}_{BA}\EBM.$$
Then it holds that
\begin{align*}
w_k^A&= \bar{\mathbf G}_Aw_k^A+\hat{\mathbf \Omega}_A\hat r_k^A+\bar{\mathbf H}_Ae_k,\\
w_k^B&= \bar{\mathbf G}_Bw_k^B+\hat{\mathbf \Omega}_B\hat r_k^B +\bar{\mathbf H}_Be_k.
\end{align*}
We then let $w_k^A=(y_k^A,u_k^A)$ and similarly for $w_k^B$ and we also let
$$\bar{\mathbf G}_A=\BBM 0&\mathbf G_A\\\Upsilon_A&0\EBM$$
and similarly for $\bar{\mathbf G}_B$. We then have two systems of the form in \eqref{eqn:yutow}, each of
which is equivalent to the form in (\ref{eqn:y-model}--\ref{eqn:u-model}). We may then use the result in
Section~\ref{sec:appendix} to deduce that they can be written as
\begin{align*}
\tilde y_k^A&=\tilde{\mathbf G}_A\tilde u_k^A+\tilde{\mathbf H}_A e_k\\
\tilde u_k^A&=\tilde \Upsilon_A\tilde y_k^A+\tilde \Omega_A\hat r_k\\
\tilde y_k^B&=\tilde{\mathbf G}_B\tilde u_k^B+\tilde{\mathbf H}_B e_k\\
\tilde u_k^B&=\tilde \Upsilon_B\tilde y_k^B+\tilde \Omega_B\hat r_k
\end{align*}
for diagonal transfer function matrices and where the other matrices are zero-one matrices. Moreover, it holds
for zero-one matrices $(\Phi_A,\Phi_B)$ that 
$$w^A_k=\Phi_A\tilde y_k^A,\quad w^B_k=\Phi_B\tilde y_k^B.$$
It then holds with $\tilde\Omega_A=\BBM \tilde \Omega_{AA}&
\tilde \Omega_{AB}\EBM$ and $\tilde\Omega_B=\BBM \tilde \Omega_{BB}&
\tilde \Omega_{BA}\EBM$ that
\begin{align*}
\tilde \Omega_A\hat r_k^A&=\tilde \Omega_{AA}r_k+\tilde \Omega_{AB}\Phi_B\tilde y_k^B\\
\tilde \Omega_B\hat r_k^B&=\tilde \Omega_{BB}r_k+\tilde \Omega_{BA}\Phi_A\tilde y_k^A.
\end{align*}
If we define $\tilde \Upsilon_{AB}=\tilde \Omega_{AB}\Phi_B$ and 
$\tilde \Upsilon_{BA}=\tilde \Omega_{BA}\Phi_A$ we obtain the 
network in Figure~\ref{fig:block-overall-partitioned}, where we do not write out the tildes. The conclusion
of this exercise is that there is no loss in generality to consider a feedback connection as in
Figure~\ref{fig:block-overall-partitioned} where $\mathbf G_A$ and 
$\mathbf G_B$ are diagonal and where the $\Upsilon$ and 
$\Omega$ matrices are zero-one matrices. 
Let $x_k$ denote the measured part of $(\mathbf y_k^A,\mathbf u_k^A)$.
All modules included in the A-part are estimated. The signals entering the
A-part from the B-part are treated as additional measured input signals. The
resulting equations have the form considered in
\cite{hansson2025identifiabilitymaximumlikelihoodestimation}.
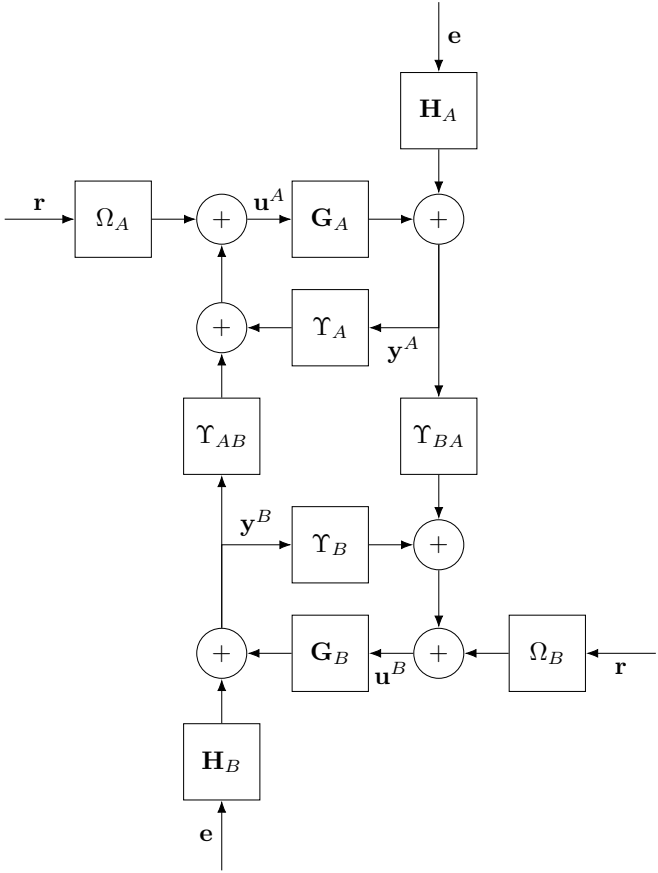
\begin{figure}[htbp]
\begin{center}
\resizebox{0.48\textwidth}{!}{\begin{tikzpicture}[auto, node distance=2cm]

\node [input, name=r_A] {};
\node [block, right of=r_A, node distance=1.5cm] (Omega_A) {$\Omega_A$};
\node [sum, right of=Omega_A, node distance=1.5cm] (sum1A) {$+$};
\node [block, right of=sum1A, node distance=1.5cm] (G_A) {$\mathbf G_A$};
\node [sum, right of=G_A, node distance=1.5cm] (sum2A) {$+$};
\node [block, above of=sum2A, node distance=1.5cm] (H_A) {$\mathbf H_A$};
\node [input, above of=H_A, node distance=1.5cm] (e_A) {};

\node [sum, below of=sum1A, node distance=1.5cm] (aux1A) {$+$};
\node [input, below of=sum2A, node distance=1.5cm] (aux2A) {};

\node [block, below of=G_A, node distance=1.5cm] (Upsilon_A) {$\Upsilon_A$};

 \draw [-Latex] (r_A) -- node {$\mathbf r$} (Omega_A); 
\draw [-Latex] (Omega_A) -- node {} (sum1A); 
\draw [-Latex] (sum1A) -- node {$\mathbf u^A$} (G_A); 
\draw [-Latex] (G_A) -- node {} (sum2A);
\draw [-Latex] (H_A) -- node {} (sum2A);
\draw [-Latex] (e_A) -- node {$\mathbf e$} (H_A); 

\draw [-Latex] (sum2A.south) to (aux2A)  -- node {$\mathbf y^A$} (Upsilon_A);
\draw [-Latex] (Upsilon_A) -- node {} (aux1A);
\draw [-Latex] (aux1A)  -- node {} (sum1A.south);


\node [block, below of=aux2A, node distance=1.5cm] (Upsilon_BA) {$\Upsilon_{BA}$};
\node [block, below of=aux1A, node distance=1.5cm] (Upsilon_AB) {$\Upsilon_{AB}$};
\node [block, below of=Upsilon_A, node distance=3cm] (Upsilon_B) {$\Upsilon_{B}$};
\node [input, below of=Upsilon_AB, node distance=1.5cm] (aux1B) {};
\node [sum, below of=Upsilon_BA, node distance=1.5cm] (sum1B) {$+$};
\node [block, below of=Upsilon_B, node distance=1.5cm] (G_B) {$\mathbf G_B$};
\node [sum, right of=G_B, node distance=1.5cm] (sum2B) {$+$};
\node [sum, left of=G_B, node distance=1.5cm] (sum3B) {$+$};
\node [block, below of=sum3B, node distance=1.5cm] (H_B) {$\mathbf H_B$};
\node [input, below of=H_B, node distance=1.5cm] (e_B) {};
\node [block, right of=sum2B, node distance=1.5cm] (Omega_B) {$\Omega_B$};
\node [input, right of=Omega_B, node distance=1.5cm] (r_B) {};

\draw [-Latex] (r_B) -- node {$\mathbf r$} (Omega_B); 
\draw [-Latex] (e_B) -- node {$\mathbf e$} (H_B); 
\draw [-Latex] (H_B) -- node {} (sum3B); 
\draw [-Latex] (G_B) -- node {} (sum3B);
\draw [-Latex] (sum2B) -- node {$\mathbf u^B$} (G_B);
\draw [-Latex] (Omega_B) -- node {} (sum2B);
\draw [-Latex] (sum1B) -- node {} (sum2B);
\draw [-Latex] (Upsilon_BA) -- node {} (sum1B);
\draw [-Latex] (Upsilon_B) -- node {} (sum1B);
\draw [-Latex] (sum3B) to (aux1B)  -- node {$\mathbf y^B$} (Upsilon_B);

\draw [-Latex] (sum3B) to (aux1B)  -- node {} (Upsilon_AB);
\draw [-Latex] (sum2A.south) to (aux2A)  -- node {} (Upsilon_BA);
\draw [-Latex] (Upsilon_AB) -- node {} (aux1A);

\end{tikzpicture}}
\end{center}
\caption{Block diagram for a dynamic network partitioned into two subnetworks.}
\centering
\label{fig:block-overall-partitioned}
\end{figure}
More specifically, the top part of the network in Figure~\ref{fig:block-overall-partitioned} 
can be cast in the form of the network in Figure~\ref{fig:block-overall}
by replacing $\Omega$ with $\bar{\Omega}=\BBM \Omega_A&I\EBM$ and $r_k$ with $\bar r_k=(r_k,\bar y^B_k)$, where 
$\bar y^B_k= \Upsilon_{AB}y^B_k$ is a signal we assume we can measure.\footnote{We can equivalently assume
that $\Upsilon_{AB}$ is known and that we can measure $y_k^B$.} 
This is visualized in Figure~\ref{fig:block-subnetwork}. 
\begin{figure}[htbp]
\begin{center}
\resizebox{0.4\textwidth}{!}{\begin{tikzpicture}[auto, node distance=2cm]

\node [input, name=r] {};
\node [block, right of=r, node distance=2.5cm] (Omega) {$\BBM \Omega_A  & I \EBM$};
\node [sum, right of=Omega, node distance=1.5cm] (sum1) {$+$};
\node [block, right of=sum1, node distance=1.5cm] (G) {$\mathbf G_A$};
\node [sum, right of=G, node distance=1.5cm] (sum2) {$+$};
\node [block, above of=sum2, node distance=1.5cm] (H) {$\mathbf H_A$};
\node [input, above of=H, node distance=1.5cm] (e) {};

\node [input, below of=sum1, node distance=1.5cm] (aux1) {};
\node [input, below of=sum2, node distance=1.5cm] (aux2) {};

\node [block, below of=G, node distance=1.5cm] (Upsilon) {$\Upsilon_A$};

 \draw [-Latex] (r) -- node {$(\mathbf r, \bar{\mathbf y}^B)$} (Omega); 
\draw [-Latex] (Omega) -- node {} (sum1); 
\draw [-Latex] (sum1) -- node {$\mathbf u^A$} (G); 
\draw [-Latex] (G) -- node {} (sum2);
\draw [-Latex] (H) -- node {} (sum2);
\draw [-Latex] (e) -- node {$\mathbf e$} (H); 

\draw [-Latex] (sum2.south) to (aux2)  -- node {$\mathbf y^A$} (Upsilon);
\draw [-Latex] (Upsilon) to (aux1)  -- node {} (sum1.south);

\end{tikzpicture}}
\end{center}
\caption{Block diagram for a dynamic subnetwork.}
\centering
\label{fig:block-subnetwork}
\end{figure}
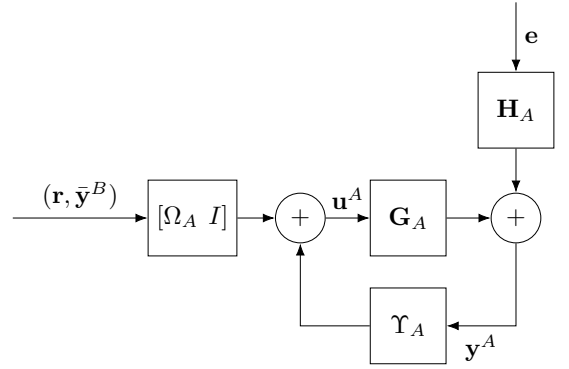
It should be stressed that $\bar y^B_k$ will in general be correlated with
$e_k$, and therefore the criterion function obtained by applying the approach
of \cite{hansson2025identifiabilitymaximumlikelihoodestimation} to a
subnetwork is not, in general, a true ML criterion. We impose the following
boundary-information condition: the component of $\bar y_k^B$ generated by
signals from A is causally reconstructible from $\{x_j:j<k\}$ and the measured
exogenous inputs available at time $k$. Thus, $\bar y_k^B$ provides no
additional information about the current A-part innovation.

It is not necessary to measure every signal from A entering B. Only a boundary
component that can affect $\bar y_k^B$ must be measured or causally
reconstructible. Measuring every signal from A entering B is a simple, but
potentially conservative, structural condition that ensures the
boundary-information condition. If this structural condition does not hold,
relevant modules can be moved from B to A and the condition checked
at the new boundary.

Enlargement only moves modules from B to A and terminates after at
most all modules in B have been included. In the worst case, estimation covers
the complete network. Enlargement is needed only to obtain the simple
structural condition; direct verification of reconstructibility may permit a
smaller A-part. The method is most useful when a small measured or
reconstructible boundary separates the target A-part from the B-part;
otherwise, the number of nuisance modules and the optimization dimension can
become large.

\cite{VanDenHof2013,Dankers2016,Weerts2020} and the present method may enlarge
the identification model beyond the original target. The former methods
enlarge or abstract the predictor model around a designated module, whereas the present method enlarges
the A-part by including original network modules until the
boundary-information condition is satisfied. These enlargements need not
contain the same transfer functions or require the same measured signals. The
present method also allows the initial A-part to contain several target
modules. If only one transfer function is the target, the initial A-part may
contain only that module, and the method applies as a special case. Unlike the
cited local predictor formulations, the resulting predictor may depend
nonlinearly on the open-loop transfer functions and can consequently require
fewer internal measurements.
In the remaining part of this
section we restrict ourselves to the case when $(\mathbf G_A,\mathbf H_A)$
defines ARMAX models. Notice that they are single-input
single-output models, since 
the transfer functions are assumed to be diagonal. 

To estimate the parameter vector $\theta$ for 
subnetwork A we consider the following 
optimization problem 
\BEQ\label{opt:ML-stationary}
\begin{aligned}
\minimize\quad&\frac{1}{2}\sum_{k=1}^N \left(\epsilon_k^T\Sigma_{\epsilon}^{-1}\epsilon_k+
\ln \det \Sigma_{\epsilon}\right)\\
\st\quad& \hat{x}_k=\mathbf W_xx_k+\mathbf W_r\bar r_k\\
&\epsilon_k=x_{k}-\hat x_{k},\quad\theta\in\Theta
\end{aligned}
\EEQ
with variables $(\theta,\epsilon,\hat x)$,
where the predictor transfer functions $\mathbf W_x$ and  $\mathbf W_r$ and 
the covariance $\Sigma_\epsilon$ are implicitly parameterized 
via an algebraic Riccati equation, cf.
\cite{hansson2025identifiabilitymaximumlikelihoodestimation}. The signal $x$
is measured in addition to $\bar r$.
In case $\bar y^B_k$ is
independent of $e_k$ it holds that this is an ML problem. This is 
however not the case in general for our problem, and then we have what in the literature is
called a PEM problem with an ML-like criterion function, cf. \cite{ljung_conver_1978}. 

We will not solve the problem as it stands, but instead
use an equivalent formulation presented in \cite{10590797} and 
based on maximum likelihood estimation for missing 
data, \cite{HANSSON20121955,Wallin02112014}. 
In this way we will not have to consider the algebraic Riccati equation as a constraint, 
which would have been difficult. However, the predictor formulation we have above will make
it possible for us to conclude that we can estimate $\mathbf G_A$ consistently. 

Consistency of the corresponding estimator for identification
of complete networks was established in \cite{hansson2025identifiabilitymaximumlikelihoodestimation}. The predictor
used here is the same Kalman filter predictor, but applied to
the subnetwork shown in Figure~\ref{fig:block-subnetwork}, where the signals entering
the subnetwork from the surrounding network are treated as
additional measured inputs. Consequently, the only difference
compared to the setting considered in \cite{hansson2025identifiabilitymaximumlikelihoodestimation} is that some of the
measured inputs originate from the surrounding feedback
network.

\begin{theorem}
Let $\hat\theta_N$ be the parameter component of a global minimizer of
\eqref{opt:ML-stationary} based on $N$ samples. Suppose that the
boundary-information condition holds and that the following conditions are satisfied:
(i) the $A$- and $B$-polynomials of the ARMAX models are coprime;
(ii) the $C$-polynomials of the ARMAX models have no zeros on the unit circle;
(iii) $(x_k,\bar r_k)$ is informative enough; and (iv) the A-part of the network is
generically network identifiable. Then, by Theorem~7 of
\cite{hansson2025identifiabilitymaximumlikelihoodestimation},
\begin{equation}
 \mathbf G_A(q;\hat\theta_N)
 \xrightarrow[N\to\infty]{\mathrm{a.s.}} \mathbf G_A(q;\theta_0)
 \label{eq:GA-consistency}
\end{equation}
where $\mathrm{a.s.}$ denotes convergence with probability one.
\end{theorem}
\begin{proof}
At the true parameter, the first two stated conditions provide the regularity
conditions for the ARMAX predictor. Strict properness of $\mathbf G(q)$
prevents the current disturbance in the
enlarged A-part from reaching $\bar y_k^B$. Diagonality of
$\bar{\mathbf H}(q)$ and $\Sigma_e$ separates the current disturbances in the
enlarged A-part and the remaining B-part. Finally, the boundary-information
condition ensures that every measured boundary signal entering A depends on A
only through information reconstructible from past measurements. Consequently,
the information used by the predictor contains no current A-part innovation,
and
\begin{equation}
 \mathbb E\!\left[\epsilon_k x_j^T\right]=0,\quad j<k,
 \qquad
 \mathbb E\!\left[\epsilon_k\bar r_j^T\right]=0,\quad j\leq k.
 \label{eq:subnetwork-orthogonality}
\end{equation}
The orthogonality relations in \eqref{eq:subnetwork-orthogonality} are the
additional condition needed in the subnetwork setting. The standard
prediction-error consistency argument, including the closed-loop case, is
given in \cite[Secs.~8.3--8.4 and 13.4--13.5]{Ljung1999}; see also
\cite{Forssell1999}. Together with the informativity and identifiability
conditions, these relations allow the consistency result in Section VII of
\cite{hansson2025identifiabilitymaximumlikelihoodestimation} to be applied to
the enlarged A-part.
\end{proof}
\begin{remark}
For the subvector $\theta_{G_A}$ parameterizing $\mathbf G_A$,
$\hat\theta_{G_A,N}\xrightarrow{\mathrm{a.s.}}\theta_{G_A,0}$ as $N\to\infty$.
Moreover, without loss of generality, it may be necessary to replace $x_k$
by suitable linearly independent linear combinations of its components in
order to satisfy Assumption~A2 of
\cite{hansson2025identifiabilitymaximumlikelihoodestimation}.
See \cite{hansson2026informativityidentifiabilityidentificationnetworks} for
informativity and generic network identifiability.
\end{remark}
The main motivation for considering the PEM formulation above is that it
permits predictors that are nonlinear in the open-loop transfer functions in
$\mathbf G_A$. Such predictors can require fewer measured signals in the
A-part than linear predictors while still yielding consistent estimates of
the parameters of $\mathbf G_A$.
It should be mentioned that, in general, we cannot estimate
the noise model. 

\section{Numerical Example}\label{sec:numerical}
\subsection{Network and measurement configurations}
We consider a seven-module network and estimate $\mathbf G^1,\mathbf G^2,\mathbf G^3$, with
$A=\{1,2,3\}$ and $B=\{4,5,6,7\}$. Its interconnection is specified by
\begin{align}
 u^1&=y^6,&u^2&=r^1+y^1+y^3,&u^3&=r^2+y^2,\notag\\
 u^4&=y^7,&u^5&=y^4,&u^6&=r^3+y^5+y^7,\notag\\
 u^7&=y^3.&&&&\label{eq:example-interconnection}
\end{align}
For this example, $\mathbf\Omega_y=\mathbf H_u=0$,
$\mathbf\Upsilon=\Upsilon$, and $\mathbf\Omega_u=\Omega$.
The interconnection matrices in \eqref{eqn:u-model} for
\eqref{eq:example-interconnection} are
\begin{equation}
\Upsilon=\begin{bmatrix}
0&0&0&0&0&1&0\\
1&0&1&0&0&0&0\\
0&1&0&0&0&0&0\\
0&0&0&0&0&0&1\\
0&0&0&1&0&0&0\\
0&0&0&0&1&0&1\\
0&0&1&0&0&0&0
\end{bmatrix},\quad
\Omega=\begin{bmatrix}
0&0&0\\1&0&0\\0&1&0\\0&0&0\\0&0&0\\0&0&1\\0&0&0
\end{bmatrix}.\label{eq:example-matrices}
\end{equation}
In the notation of the preceding section, the incoming signal is
$\bar y^B=\Upsilon_{AB}y^B=(y^6,0,0)$. The external input $r^3$ does not
enter $A$ directly. Omitting these unused and zero components, the measured
input signal for estimation is $(r^1,r^2,y^6)$, and
\begin{equation}
\begin{aligned}
 \Upsilon_A&=\begin{bmatrix}0&0&0\\1&0&1\\0&1&0\end{bmatrix},\\
 \Omega_A r_k+\Upsilon_{AB}y_k^B
 &=\begin{bmatrix}0&0&1\\1&0&0\\0&1&0\end{bmatrix}
 \begin{bmatrix}r_k^1\\r_k^2\\y_k^6\end{bmatrix}.
\end{aligned}\label{eq:example-subnetwork-matrices}
\end{equation}
The measured output vector $x_k$ is $y_k^3$, $(y_k^1,y_k^3)$, or
$(y_k^1,y_k^2,y_k^3)$ in the three configurations considered; the same
measured inputs are available in every case.
The signal from A to B that can return to A is $y_k^3$, which is included in
every choice of $x_k$. Hence the selected A-part satisfies the simple
structural sufficient condition, and no enlargement is needed in this
example.

For the diagram, write $v^i=\mathbf H^i e^i$ for each module disturbance.
\begin{figure*}[t]
\centering
\resizebox{0.82\textwidth}{!}{\begin{tikzpicture}[auto, node distance=2cm]
    \node [input] (sum11) {$+$};
    \node [block, right of=sum11, node distance=3cm] (G1) {$\mathbf G^1$};
    \node [sum, right of=G1, node distance=2cm] (sum1) {$+$};
    \node [input, above of=sum1, node distance=1.2cm] (e1) {};

    \node [sum, right of=sum1, node distance=2cm] (sum12) {$+$};
    \node [block, right of=sum12, node distance=3cm] (G2) {$\mathbf G^2$};
    \node [sum, right of=G2, node distance=2cm] (sum2) {$+$};
    \node [input, below of=sum2, node distance=1.2cm] (e2) {};
    \node [input, below of=sum12, node distance=1.2cm] (r2) {};
    
    \node [sum, right of=sum2, node distance=2cm] (sum13) {$+$};
    \node [input, below of=sum13, node distance=1.2cm] (r3) {};
    \node [block, right of=sum13, node distance=2cm] (G3) {$\mathbf G^3$};
    \node [sum, right of=G3, node distance=2cm] (sum3) {$+$};
    \node [input, below of=sum3, node distance=1.2cm] (e3) {};
    \node [input, right of=sum3] (y3) {};


    \node [input, above of=sum3 ,node distance=1cm] (aux3) {};
    \node [input, above of=sum12, node distance=1cm] (aux4) {};


    \node [input, above of=sum3 ,node distance=2.5cm] (auxd1) {};
    \node [block, left of=auxd1, node distance=3cm] (G7) {$\mathbf G^7$};
    \node [sum, left of=G7, node distance=2cm] (sum7) {$+$};
    \node [input, below of=sum7, node distance=1.2cm] (e7) {};

    \node [sum, above of=sum11, node distance=2.5cm] (sum6) {$+$};
    \node [block, right of=sum6, node distance=1.7cm] (G6) {$\mathbf G^6$};
    \node [sum, right of=G6, node distance=2.5cm] (sumu6) {$+$};
    \node [input, below of=sumu6, node distance=1.2cm] (r4) {};
    \node [input, above of=sum6, node distance=1.2cm] (e6) {};


    \node [input, above of=sum7 ,node distance=1.5cm] (auxB1) {};
    \node [block, left of=auxB1, node distance=2.5cm] (G4) {$\mathbf G^4$};
    \node [sum, left of=G4, node distance=1.7cm] (sum4) {$+$};
    \node [sum, above of=sumu6, node distance=1.5cm] (sum5) {$+$};
    \node [block, right of=sum5, node distance=1.7cm] (G5) {$\mathbf G^5$};
    \node [input, above of=sum4, node distance=1.2cm] (e4) {};
    \node [input, above of=sum5, node distance=1.2cm] (e5) {};

    \draw [-Latex] (e1) -- node {$v^1$} (sum1);
    \draw [-Latex] (G1) -- node {} (sum1);
    \draw [-Latex] (sum1) -- node {$y^1$} (sum12);

    \draw [-Latex] (r2) -- node {$r^1$} (sum12);
    \draw [-Latex] (sum12) -- node {$u^2$} (G2);
    \draw [-Latex] (G2) -- node {} (sum2);
    \draw [-Latex] (e2) -- node {$v^2$} (sum2);
    \draw [-Latex] (sum2) -- node {$y^2$} (sum13);

    \draw [-Latex] (r3) -- node {$r^2$} (sum13);
    \draw [-Latex] (sum13) -- node {$u^3$} (G3);
    \draw [-Latex] (G3) -- node {} (sum3);
    \draw [-Latex] (e3) -- node {$v^3$} (sum3);
    \draw [-Latex] (sum3) to (aux3) to (aux4) -- node {$y^3$} (sum12);

    \draw [-Latex] (sum3) to (auxd1) -- node {$u^7 = y^3$} (G7);
    \draw [-Latex] (G7) -- node {} (sum7);
    \draw [-Latex] (e7) -- node {$v^7$} (sum7);

    \draw [-Latex] (G6) -- node {} (sum6);
    \draw [-Latex] (sumu6) -- node {$u^6$} (G6);
    \draw [-Latex] (r4) -- node {$r^3$} (sumu6);
    \draw [-Latex] (sum6) to (sum11) -- node {$u^1 = y^6$} (G1);
    \draw [-Latex] (e6) -- node {$v^6$} (sum6);
    \draw [-Latex] (sum7) -- node {$y^7$} (sumu6);


    \draw [-Latex] (sum7) to (auxB1) -- node {$u^4 = y^7$} (G4);
    \draw [-Latex] (G4) -- node {} (sum4);
    \draw [-Latex] (G5) -- node {} (sum5);
    \draw [-Latex] (e4) -- node {$v^4$} (sum4);
    \draw [-Latex] (e5) -- node {$v^5$} (sum5);
    \draw [-Latex] (sum5) -- node {$y^5$} (sumu6);
    \draw [-Latex] (sum4) -- node {$u^5 = y^4$} (G5);

\begin{scope}[on background layer]
\node[draw,dashed,rounded corners,inner sep=7pt,fit=(G1)(sum1)(sum12)(G2)(sum2)(sum13)(G3)(sum3)(e1)(e2)(e3)(r2)(r3),label={[font=\large]below:Target subnetwork $A$}] {};
\end{scope}
\end{tikzpicture}}
\caption{Example network, with target modules $A=\{1,2,3\}$ enclosed by the
dashed line and remaining modules in $B$. The signal $y^6$ is measured
in every experiment. Internal measurements are $y^3$, $(y^1,y^3)$, or
$(y^1,y^2,y^3)$.}
\label{fig:network_num_example}
\end{figure*}
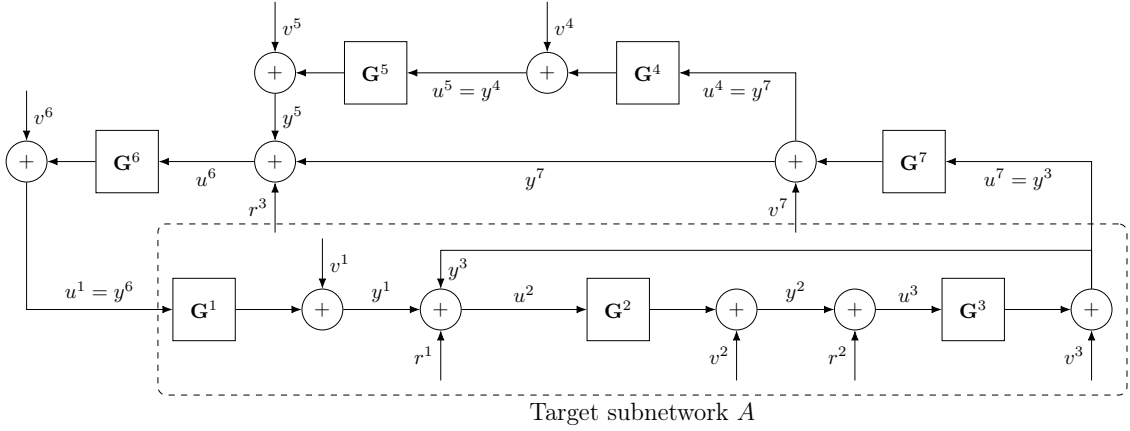

All modules have order two, with transfer functions
\begin{equation}
 \mathbf G^i(q)=\frac{b_1^i q+b_2^i}{q^2+a_1^i q+a_2^i},\qquad
 \mathbf H^i(q)=\frac{q^2+c_1^i q+c_2^i}{q^2+a_1^i q+a_2^i}.
 \label{eq:numerical-modules}
\end{equation}
The plant coefficients are $a_1^i,a_2^i,b_1^i,b_2^i$; the coefficients
$c_1^i,c_2^i$ describe the noise numerator, and $\lambda^i$ is the
variance of $e^i_k$. 

\begin{table}[t]
\centering
\caption{Data-generating coefficients; $\lambda^i=0.01i$.}\label{tab:coefficients}
\setlength{\tabcolsep}{4pt}
\begin{tabular}{crrrrrr}
\hline
$i$&$a_1^i$&$a_2^i$&$b_1^i$&$b_2^i$&$c_1^i$&$c_2^i$\\
\hline
1&1&0.25&0.3&0.16&0.3&-0.01\\
2&-0.8&0.15&0.8&-0.3&-0.8&0.2\\
3&0.45&-0.13&-0.4&-0.25&-0.02&-0.8\\
4&-0.45&-0.1&-2&0.4&-0.15&-0.07\\
5&0.1&-0.4&2.2&2&-0.6&-0.05\\
6&-0.2&-0.15&0.15&0.05&1&0.15\\
7&0.5&0.05&1&0.2&0.15&-0.7\\
\hline
\end{tabular}
\end{table}

\subsection{Algebraic recovery and identifiability}\label{sec:algebraic-recovery}
For the coefficients in Table~\ref{tab:coefficients}, the poles
of the full network and of subnetwork $A$, with $(r^1,r^2,y^6)$ treated
as external inputs, have absolute values less than $0.91850$ and $0.64341$, respectively. All seven plant modules are stable and strictly proper.
Eliminating the unmeasured outputs $y^1$ and $y^2$ gives
\begin{align}
 \Delta y^3={}&\mathbf G^3\mathbf G^2\mathbf G^1 y^6
 +\mathbf G^3\mathbf G^2r^1+\mathbf G^3r^2\notag\\
 &+\mathbf G^3\mathbf G^2\mathbf H^1e^1
 +\mathbf G^3\mathbf H^2e^2+\mathbf H^3e^3,
 \label{eq:eliminate}
\end{align}
where $\Delta=1-\mathbf G^3\mathbf G^2$ and transfer operators act on
the associated signals. In the input ordering $(y^6,r^1,r^2)$, the transfer
matrix from measured inputs to $y^3$ has entries
\begin{equation}
 T_1=\frac{\mathbf G^3\mathbf G^2\mathbf G^1}{\Delta},\qquad
 T_2=\frac{\mathbf G^3\mathbf G^2}{\Delta},\qquad
 T_3=\frac{\mathbf G^3}{\Delta}.\label{eq:closed-loop-map}
\end{equation}
The target modules are recovered by
\begin{equation}
 \mathbf G^1=T_1/T_2,\qquad \mathbf G^2=T_2/T_3,\qquad
 \mathbf G^3=T_3/(1+T_2).\label{eq:recovery}
\end{equation}
Equations~\eqref{eq:recovery} establish generic uniqueness, or generic network 
identifiability from
$(T_1,T_2,T_3)$ whenever the denominators are not identically zero.
For coprime plant polynomials with monic denominators of the prescribed
orders, the transfer functions in \eqref{eq:recovery} also determine the
plant coefficients. Thus the three target plants can be recovered from
the transfer functions between the measured inputs and $y^3$, even when
$y^1$ and $y^2$ are unmeasured.

Direct linear PEM for all three open-loop transfer functions requires measuring
$(y^1,y^2,y^3)$ together with the measured signals
$(y^6,r^1,r^2)$. This gives the equations
\begin{equation}
\begin{aligned}
 y^1&=\mathbf G^1y^6+\mathbf H^1e^1,\\
 y^2&=\mathbf G^2(r^1+y^1+y^3)+\mathbf H^2e^2,\\
 y^3&=\mathbf G^3(r^2+y^2)+\mathbf H^3e^3.
\end{aligned}\label{eq:direct-pem}
\end{equation}
In \eqref{eq:direct-pem}, every open-loop transfer function has a fully measured input.
With fewer internal measurements, direct PEM will be nonlinear in the  
open-loop transfer functions.

\subsection{Estimation, validation, and baseline}
The estimation data contain $N=500$ samples with independent Gaussian
disturbances of variances $\lambda^i=0.01i$. The external signals are random
binary sequences taking values $-1$ and $1$, generated using MATLAB's
\texttt{idinput} with the \texttt{rbs} option. For parameter statistics,
100 disturbance realizations use the same external-input data. Thus these
statistics describe variability conditional on that excitation data.

We do not impose the algebraic Riccati equation as an optimization
constraint. Instead, we use the equivalent missing-data formulation of
\cite{10590797}, based on \cite{HANSSON20121955,Wallin02112014}.
The implementation uses three initialization stages: autoregressive estimation
with exogenous inputs (ARX) and a
common disturbance variance, ARMAX estimation with that common variance, and
ARMAX estimation with separate variances. It uses a trust-region method from SciPy's \texttt{optimize} package, with derivatives computed using JAX \cite{jax}, double precision, default SciPy stopping
tolerances, and at most 1000 iterations per stage. The second and third stages are initialized
with the estimate from the preceding stage. Positive noise variances are
parameterized as $\lambda^i=10^{-2\rho_i}$ with unconstrained $\rho_i$.

The baseline is direct PEM, implemented through separate
ARMAX estimation in MATLAB R2024a using the System Identification Toolbox
\cite{ljung1995system}. The three equations in \eqref{eq:direct-pem} are fitted
separately using \texttt{armax} with orders
$[n_a,n_b,n_c,n_k]=[2,2,2,1]$ and
\texttt{SearchOptions.\allowbreak Tolerance}$=10^{-5}$.
Here $n_a,n_c$ are polynomial orders, $n_b$ counts input coefficients, and
$n_k$ is the input delay. With sparse measurement, not all
required input/output pairs are available; this baseline is therefore used
only when all three outputs are measured.

For validation, a separate 500-sample data record supplies $r^1,r^2,y^6$ to the
estimated subnetwork. Plant-only simulations $\hat y^i$, without new internal
disturbances, are compared with noisy reference outputs $y^{i,\mathrm{ref}}$ using
\begin{equation}
 F_i=100\left(1-\frac{\|\hat y^i-y^{i,\mathrm{ref}}\|_2}
 {\|y^{i,\mathrm{ref}}-\mathbf 1\bar y^{i,\mathrm{ref}}\|_2}\right),
 \label{eq:fit}
\end{equation}
where $\bar y^{i,\mathrm{ref}}$ is the sample mean and $\mathbf 1$ is the
500-dimensional vector of ones. These individual simulation fits use
measured inputs, not one-step predictions or free full-network simulation.
Withheld outputs are used only for scoring. Both methods use the same
500 validation samples without burn-in removal.

For the Monte Carlo statistics, define the target coefficient vector as
$\alpha_A=(a_1^1,a_2^1,b_1^1,b_2^1,\ldots,a_1^3,a_2^3,b_1^3,b_2^3)
\in\mathbb R^{12}$. Let $\hat\alpha^{(\ell)}$ denote its estimate in run $\ell$.
Let $\alpha_{A,0}$ be its true value. For $R=100$ runs, define the sample
mean $\bar\alpha$, empirical covariance $\widehat C$, squared bias $B^2$,
and mean-square error (MSE) by
\begin{align}
 \bar\alpha&=R^{-1}\sum_{\ell=1}^R\hat\alpha^{(\ell)},\notag\\
 \widehat C&=R^{-1}\sum_{\ell=1}^R(\hat\alpha^{(\ell)}-\bar\alpha)
 (\hat\alpha^{(\ell)}-\bar\alpha)^T,\notag\\
 B^2&=\|\bar\alpha-\alpha_{A,0}\|_2^2,\qquad
 \mathrm{MSE}=\operatorname{tr}\widehat C+B^2.\label{eq:statistics}
\end{align}
The divisor $R$ makes \eqref{eq:statistics} exact. All saved estimates are
included, without stability-based exclusion. Code and results are
public.\footnote{\url{https://github.com/Jvictormata/mle_sub_nets}}

\begin{table*}[t]
\centering
\caption{Individual simulation fits and 100-run parameter statistics;
measured inputs: $y^6,r^1,r^2$.}
\label{tab:table_PEM_vs_MLE}
\begin{tabular}{llrrrrrr}
\hline
Internal outputs&Method&$F_1$&$F_2$&$F_3$&$\operatorname{tr}\widehat C$&$B^2$&MSE\\
\hline
$y^3$&Proposed PEM&56.58&75.89&60.36&0.3764&1.3776&1.7539\\
$y^1,y^3$&Proposed PEM&57.10&75.84&60.21&0.2442&1.2648&1.5090\\
$y^1,y^2,y^3$&Proposed PEM&57.13&76.86&60.43&0.8501&0.1974&1.0475\\
$y^1,y^2,y^3$&Direct PEM (linear in $\mathbf G$)&57.13&76.86&60.43&0.8717&0.1400&1.0117\\
\hline
\end{tabular}
\end{table*}

\subsection{Results and discussion}
For the network in Figure~\ref{fig:network_num_example}, the proposed PEM
and direct PEM give similar results with the same measured signals.
Its main advantage is that all three open-loop transfer functions can be
recovered using fewer internal measurements, with only a small decrease in
the fit values: measuring only $y^3$ changes each fit by less than one
percentage point relative to complete internal measurement.

A larger parameter bias is observed when fewer signals are measured.
The squared bias increases from $0.1974$ with all three outputs to $1.3776$
with $y^3$ alone, while the parameter MSE increases from $1.0475$ to $1.7539$.
The objective in \eqref{opt:ML-stationary} is nonconvex, so convergence to
local minima can contribute to the observed parameter bias.

\section{Conclusions}
This paper has shown that a class of subnetwork identification problems can be reformulated as complete-network identification problems through augmentation of the measured external signals.
By augmenting the measured exogenous signals with measured boundary signals entering the subnetwork from the surrounding network, the subnetwork identification problem can be cast in a form closely related to the identification of a complete network.

Consistency requires the A-generated component of each measured boundary
signal entering the A-part to be reconstructible from the measured history.
Measuring every signal from A entering B is a simple sufficient condition,
but it is not necessary. The A-part may be enlarged to satisfy this structural
condition, although reconstructibility can permit a smaller subnetwork. The
same predictor construction as in
\cite{hansson2025identifiabilitymaximumlikelihoodestimation} can then be
applied to the enlarged A-part.

An important advantage of the approach is that it allows the use of predictors that are nonlinear in the open-loop transfer functions. This makes it possible to obtain consistent estimates using fewer measured signals than what is typically required by PEM methods based on predictors that are linear in the open-loop transfer functions.

Future work will focus on numerical aspects of the optimization problem and further investigations of measurement requirements for identification of subnetworks.
\section{Appendix}\label{sec:appendix}
Consider \eqref{eqn:y-model}--\eqref{eqn:u-model}, and assume that
$m_e\leq n$. For $p\in\{n,m_r\}$, define
\begin{equation}
 J_p=\mathbbm 1_n\otimes I_p\in\mathbb R^{np\times p},
 \qquad
 S_p=I_n\otimes\mathbbm 1_p^T\in\mathbb R^{n\times np},
 \label{eq:appendix-replication-summation}
\end{equation}
where $\mathbbm 1_p$ is the $p$-dimensional vector of ones. For a transfer
matrix $\mathbf A(q)\in\mathbb R^{n\times p}(q)$, let
$\operatorname{vec}_{\rm r}(\mathbf A)$ denote the vector obtained by stacking
the entries of $\mathbf A$ row by row, and define
\begin{equation}
 D_{\mathbf A}(q)
 =\diag\!\left(\operatorname{vec}_{\rm r}(\mathbf A(q))\right)
 \in\mathbb R^{np\times np}(q).
 \label{eq:appendix-diagonalization}
\end{equation}
The definitions in \eqref{eq:appendix-replication-summation}--
\eqref{eq:appendix-diagonalization} give
\begin{equation}
 \mathbf A(q)z_k=S_pD_{\mathbf A}(q)J_pz_k,
 \qquad z_k\in\mathbb R^p.
 \label{eq:appendix-factorization}
\end{equation}

Let $d=n^2+nm_r$ and define
\begin{equation}
 \begin{aligned}
 e_k^p&=\BBM e_k\\0_{n-m_e}\EBM\in\mathbb R^n,\\
 \mathbf H_y^p&=\BBM\mathbf H_y&0_{n\times(n-m_e)}\EBM,\\
 \mathbf H_u^p&=\BBM\mathbf H_u&0_{n\times(n-m_e)}\EBM.
 \end{aligned}
 \label{eq:appendix-noise-padding}
\end{equation}
Then $\mathbf H_y^pe_k^p=\mathbf H_ye_k$ and
$\mathbf H_u^pe_k^p=\mathbf H_ue_k$. Define
\begin{align}
 \tilde{\mathbf G}
 &=\blkdiag\!\left(
 D_{\mathbf G},D_{\mathbf\Omega_y},
 D_{\mathbf\Upsilon},D_{\mathbf\Omega_u}\right)
 \in\mathbb R^{2d\times2d}(q),
 \label{eq:appendix-Gtilde}\\
 \tilde{\mathbf H}
 &=\blkdiag\!\left(D_{\mathbf H_y^p},0_{nm_r\times nm_r},\right.\notag\\
 &\hspace{4.6em}\left.
 D_{\mathbf H_u^p},0_{nm_r\times nm_r}\right)
 \in\mathbb R^{2d\times2d}(q),
 \label{eq:appendix-Htilde}\\
 \tilde e_k
 &=\BBM
 J_ne_k^p\\0_{nm_r}\\J_ne_k^p\\0_{nm_r}
 \EBM\in\mathbb R^{2d}.
 \label{eq:appendix-etilde}
\end{align}
Both $\tilde{\mathbf G}$ and $\tilde{\mathbf H}$ are diagonal transfer
matrices.

Define
\begin{equation}
 P=\BBM S_n&S_{m_r}\EBM\in\mathbb R^{n\times d},
 \qquad K=J_nP\in\mathbb R^{n^2\times d}.
 \label{eq:appendix-PK}
\end{equation}
Using block-row dimensions $n^2,nm_r,n^2,nm_r$ and two block-column
dimensions $d,d$, let
\begin{align}
 \tilde\Upsilon
 &=\BBM
 0_{n^2\times d}&K\\
 0_{nm_r\times d}&0_{nm_r\times d}\\
 K&0_{n^2\times d}\\
 0_{nm_r\times d}&0_{nm_r\times d}
 \EBM\in\mathbb R^{2d\times2d},
 \label{eq:appendix-Upsilon-tilde}\\
 \tilde\Omega
 &=\BBM
 0_{n^2\times m_r}\\J_{m_r}\\
 0_{n^2\times m_r}\\J_{m_r}
 \EBM\in\mathbb R^{2d\times m_r}.
 \label{eq:appendix-Omega-tilde}
\end{align}
The expanded network is
\begin{align}
 \tilde y_k
 &=\tilde{\mathbf G}(q)\tilde u_k
   +\tilde{\mathbf H}(q)\tilde e_k,
 \label{eq:appendix-expanded-output}\\
 \tilde u_k
 &=\tilde\Upsilon\tilde y_k+\tilde\Omega r_k.
 \label{eq:appendix-expanded-input}
\end{align}
Finally, with
\begin{equation}
 \Psi=\blkdiag(P,P)\in\mathbb R^{2n\times2d},
 \label{eq:appendix-Phi}
\end{equation}
equation \eqref{eq:appendix-factorization} gives
\begin{equation}
 \BBM y_k\\u_k\EBM=\Psi\tilde y_k.
 \label{eq:appendix-reconstruction}
\end{equation}
Thus, \eqref{eqn:y-model}--\eqref{eqn:u-model} and
\eqref{eq:appendix-expanded-output}--\eqref{eq:appendix-reconstruction}
generate the same signals $y_k$ and $u_k$. The expanded description has
diagonal plant and noise transfer matrices and is a Box--Jenkins description.
If a component of $\tilde y_k$ is identically zero, that component and the
corresponding rows and columns can be removed.
\end{document}